\documentclass[10pt]{IEEEtran}
\usepackage{hyperref,graphicx}
\usepackage{amsmath}
\usepackage{amssymb}
\usepackage{amsthm}
\usepackage{bm} 
\usepackage{tabularx, booktabs} 
\usepackage{multirow}
\usepackage{booktabs}
\usepackage{cite}
\usepackage{xcolor}
\usepackage{lipsum}
\newcommand{\bbR}{\mathbb{R}}
\newcommand{\bbC}{\mathbb{C}}
\newcommand{\bbH}{\mathbb{H}}
\newcommand{\bbCi}{\mathbb{C}_{\bm i}}
\newcommand{\bbCj}{\mathbb{C}_{\bmj}}
\newcommand{\dt}{\mathrm{d}t}

\newcommand{\dnu}{\mathrm{d}\nu}
\newcommand{\bmj}{\bm j}
\newcommand{\bmi}{\bm i}
\newcommand{\bmk}{\bm k}
\newcommand{\bmmu}{\bm \mu}

\renewcommand{\epsilon}{\varepsilon}
\renewcommand{\phi}{\varphi}
\newcommand{\intinf}{\int_{-\infty}^{+\infty}}

\newcommand{\ip}[1]{\left\langle #1 \right\rangle}

\newcommand{\ie}{\emph{i.e. }}
\newcommand{\Expe}[1]{\mathbf{E}\left\lbrace #1 \right\rbrace}

\newcommand{\conji}[1]{#1^{*\bmi}}

\newcommand{\involi}[1]{\overline{#1}^{\bmi}}
\newcommand{\involj}[1]{\overline{#1}^{\bmj}}

\newcommand{\polarmodj}[1]{#1\bmj\overline{#1}}
\newcommand{\Span}[1]{\mathrm{span}\left\lbrace #1\right\rbrace}

\newcommand{\ve}[1]{\mathbf{#1}}
\newcommand{\mat}[1]{\underline{\mathbf{#1}}}

\newcommand{\Utwo}{\mathsf{U}(2)}
\newcommand{\SUtwo}{\mathsf{SU}(2)}
\newcommand{\matrixspace}[2]{{#1}^{#2 \times #2}}

\newtheorem{proposition}{Proposition}
\newtheorem{lemma}{Lemma}

\begin{document}

\title{A complete framework for linear filtering\\
of bivariate signals}

\author{Julien~Flamant,~\IEEEmembership{Student Member,~IEEE,}
Pierre~Chainais,~\IEEEmembership{Senior~Member,~IEEE,} and~Nicolas~Le Bihan
\thanks{J. Flamant and P. Chainais are with Univ. Lille, CNRS, Centrale Lille, UMR 9189 - CRIStAL - Centre de Recherche en Informatique Signal et Automatique de Lille, 59000 Lille, France. N. Le Bihan is with CNRS/GIPSA-Lab, 11 rue des math\'{e}matiques, Domaine Universitaire, BP 46, 38402 Saint Martin d'H\`{e}res cedex, France. Part of this work has been funded by the CNRS, GDR ISIS, within the SUNSTAR interdisciplinary research program.}}%

\maketitle

\begin{abstract}
  A complete framework for the linear time-invariant (LTI) filtering theory of bivariate signals is proposed based on a tailored quaternion Fourier transform.
  This framework features a direct description of LTI filters in terms of their eigenproperties enabling compact calculus and physically interpretable filtering relations in the frequency domain.
  The design of filters exhibiting fondamental properties of polarization optics (birefringence, diattenuation) is straightforward.
  It yields an efficient spectral synthesis method and new insights on Wiener filtering for bivariate signals with prescribed frequency-dependent polarization properties.
  This generic framework facilitates original descriptions of bivariate signals in two components with specific geometric or statistical properties.
  Numerical experiments support our theoretical analysis and illustrate the relevance of the approach on synthetic data.
\end{abstract}


\section{Introduction}
\label{sec:introduction}

\IEEEPARstart{B}{ivariate} signals appear in numerous physical areas such as optics \cite{erkmen2006optical}, oceanography \cite{gonella1972rotary}, geophysics
 \cite{samson1983pure,Roueff2006} or EEG analysis \cite{sakkalis2011review}.
A bivariate signal $x(t)$ is usually resolved into orthogonal components corresponding to real-valued signals $x_1(t)$ and $x_2(t)$.
Then $x(t)$ can be expressed either in vector form $x(t) = [x_1(t)\: x_2(t)]^T$ or as the complex valued signal $x(t) = x_1(t) + \bmi x_2(t)$.
Benefits of each representation have been reviewed recently \cite{Sykulski2017}.

Linear time-invariant (LTI) filtering theory is a cornerstone of signal processing.
Its extension to the case of bivariate signals depends on the chosen representation -- vector or complex form.
The use of the complex representation $x(t) = x_1(t) + \bmi x_2(t)$ leads to the concept of \emph{widely linear filtering} \cite{Picinbono1995,Sykulski2016,Mandic2009,Schreier,Schreier2003}, meaning that the signal $x(t)$ and its conjugate $\overline{x(t)}$ are in general filtered differently.
While the use of the complex representation is often advocated for in the signal processing literature \cite{Schreier,Picinbono1997a}, the use of the vector form $x(t) = [x_1(t)\: x_2(t)]^T$ is more common in physical sciences, \emph{e.g.} polarization optics \cite{born2000principles,J.J.Gil2016}.
The vector $x(t)$ is usually replaced by its analytic signal version -- the so-called \emph{Jones vector}.
LTI filters are then represented in the spectral domain by $2\times 2$ complex matrices called \emph{Jones matrices}.
These matrices describe optical elements or media with fondamental optical properties such as \emph{birefringence} and \emph{diattenuation}.
See \emph{e.g.} \cite{J.J.Gil2007} for a review of the \emph{Jones formalism}.

A complete framework for LTI filtering of bivariate signals should exhibit some desirable properties:
\emph{(i)} a description of bivariate signals as single algebraic objects for simple calculations (in contrast with \emph{e.g.} rotary components \cite{walden2013rotary}), \emph{(ii)} a convenient duality between time and frequency to define easily interpretable Fourier representations, \emph{(iii)} a simple representation of LTI filters in terms of their main properties, such as eigenvectors and eigenvalues (in contrast with \emph{e.g.} Jones matrices or widely linear filters), and
\emph{(iv)} a fast implementation, \emph{e.g.} relying on FFT.
As noticed, existing approaches do not fullfill these properties all at once.

We have recently introduced a powerful alternative approach to bivariate signal processing \cite{Flamant2017timefrequency,flamant2017spectral} using a tailored quaternion Fourier transform (QFT).
The proposed framework exhibits an unifying structure by directly connecting usual physical quantities from polarization to well-defined mathematical (quaternion-valued) quantities such as spectral densities, covariances, time-frequency representations, etc.
It provides at no extra cost an elegant, compact and insightful calculus which highlights the geometric treatment of polarization states.
Note that first attempts in this direction root in optics \cite{Richartz1949,Whitney1971,Pellat-finet1984,Tudor2010a,Tudor2010b} that provide a clear geometric formulation of Jones formalism.
However its generic use for bivariate signal processing is hindered by ignoring phase terms, assuming monochromatism and unpractical implementation.

The QFT framework enables an efficient description of LTI filters and overcomes the limitations of previous approaches by answering all the desirable requirements mentioned above.
In the proposed representation LTI filters are explicitly
given in the spectral domain in terms of their eigenproperties.
It provides clear and economical expressions.
The interaction between LTI filters and bivariate signals is then easy to interpret or prescribe.
It directly relates to fondamental properties of optical media known as \emph{birefringence} and \emph{diattenuation}.
This complete framework provides a new interpretable and generic approach to standard signal processing operations such as spectral synthesis and Wiener filtering for instance.
Moreover it makes natural various original descriptions of bivariate signals in two components with specific geometric or statistical properties.

This paper is organized as follows.
In Section \ref{sec:preliminaries} we gather useful properties of the QFT.
Based on a usual decomposition \cite{J.J.Gil2007,J.J.Gil2016} which separates LTI filters into \emph{unitary} and \emph{Hermitian} ones, Section \ref{sec:filters_for_bivariate_signals} presents a thorough study of each family in the QFT domain.
Section \ref{sec:applications} presents practical applications of those filters: usual ones (spectral synthesis, Wiener filtering) and original decompositions of bivariate signals into two components with prescribed properties.
Section \ref{sec:conclusion} gathers concluding remarks.
Detailed calculations are remitted to appendices.
For the sake of reproducibility, an implementation of the QFT framework along with tools presented in this paper will be available through the open-source Python companion toolbox
\texttt{BiSPy}\footnote{documentation available at \texttt{\url{https://bispy.readthedocs.io/}}}.

\section{Background}
\label{sec:preliminaries}
Section \ref{sub:quaternions} and Section \ref{sub:quaternion_fourier_transform}
present two key ingredients of this work: quaternions and the quaternion Fourier transform.
Section \ref{sub:quaternion_spectral_density} introduces the quaternion spectral density of a bivariate signal, a fundamental quantity that allows numerous physical and geometrical interpretations.

\subsection{Quaternions} 
\label{sub:quaternions}

Quaternions form a four dimensional algebra denoted $\bbH$ and with canonical basis $\lbrace 1, \bmi, \bmj, \bmk\rbrace$, where $\bmi, \bmj, \bmk$ are imaginary units $\bmi^2 = \bmj^2 = \bmk^2 = -1$ such that
\begin{equation}
    \bmi\bmj = \bmk, \: \bmi\bmj = -\bmj\bmi, \: \bmi\bmj\bmk = -1.
\end{equation}
Importantly, like matrix product quaternion multiplication is noncommutative, \ie in general for $p, q \in \bbH$ one has $pq \neq qp$.
Any quaternion $q\in \bbH$ can be written as
\begin{equation}
    q = a + b\bmi + c\bmj + d\bmk
\end{equation}
where $a, b, c, d \in \bbR$. The \emph{scalar} or \emph{real} part of $q$ is $\mathcal{S}(q) = a \in \bbR$ and its \emph{vector} or \emph{imaginary} part is $\mathcal{V}(q) = q - \mathcal{S}(q) \in \Span{\bmi, \bmj, \bmk}$. When $\mathcal{S}(q) = 0$, $q$ is said to be pure.
The quaternion conjugate of $q$ is $\overline{q} = \mathcal{S}(q) - \mathcal{V}(q)$.
Its modulus is $\vert q \vert^2 = q\overline{q} = \overline{q}q = a^2+b^2+c^2+d^2$.
Involutions with respect to $\bmi, \bmj, \bmk$ are defined by $\overline{q}^{\bmi} = - \bmi q \bmi, \: \overline{q}^{\bmj} = - \bmj q \bmj, \: \overline{q}^{\bmk} = - \bmk q \bmk$.
Involutions somehow extend the notions of complex conjugation as they represent reflections, \emph{e.g.} $\overline{q}^{\bmi} = a + b\bmi - c\bmj -d\bmk$.

Quaternions generalize naturally complex numbers.
Concepts such as imaginary units, polar forms extend nicely.
For instance $\bbCj = \Span{1, \bmj}$ or $\bbCi = \Span{1, \bmi}$ are complex subfields of $\bbH$ isomorphic to $\bbC$. As a result, given a pure unit quaternion $\bmmu$ such that $\bmmu^2 = -1$ and $\theta \in \bbR$, one gets $\exp(\bmmu\theta) = \cos\theta + \bmmu\sin\theta$.

As it is essential to our analysis we mention another property of quaternions.
Any quaternion can be represented as a pair of complex numbers.
Let $q = q_1 + \bmi q_2, \: q_1, q_2 \in \bbCj$.
The vector representation of $q$ is the 2-dimensional complex vector $\ve{q} = [q_1, q_2]^T \in \matrixspace{\bbC}{2}_{\bmj}$.
For more about quaternions, the reader is referred to dedicated textbooks \emph{e.g.} \cite{conway2003quaternions}.

\subsection{Quaternion Fourier transform} 
\label{sub:quaternion_fourier_transform}
Several Quaternion Fourier transforms have been proposed so far, see \cite{Hitzer2013} for a review.
We briefly survey the Quaternion Fourier Transform (QFT) first introduced in \cite{Bihan2014} and further studied in \cite{Flamant2017timefrequency}. Recent works \cite{Flamant2017timefrequency,flamant2017spectral} have demonstrated the relevance of this QFT to process bivariate signals.
In particular the QFT decomposes directly bivariate signals into a sum of polarized monochromatic signals.
It also allows novel, natural and direct interpretation of polarization features for bivariate signals.

A bivariate signal written as a $\bbCi$-valued signal reads $x(t) = x_1(t) + \bmi x_2(t)$, where $x_1, x_2$ are real signals.
Suppose for now that $x(t)$ is deterministic.
The QFT of $x(t)$ is then
\begin{equation}\label{eq:definitionQFT}
    X(\nu) \triangleq \intinf x(t)e^{-\bmj2\pi\nu t}\dt = X_1(\nu) + \bmi X_2(\nu) \in \bbH.
\end{equation}
where $X_1, X_2$ are the standard Fourier transform (FT) of $x_1, x_2$, taken as $\bbCj$-complex valued. The inverse QFT is given by
\begin{equation}\label{eq:inverseQFT}
     x(t) = \intinf X(\nu)e^{\bmj2\pi\nu t}\mathrm{d}\nu.
\end{equation}
The QFT (\ref{eq:definitionQFT}) is very similar to the usual FT where the axis $\bmi$ of the FT has simply been replaced by $\bmj$. Importantly, the exponential kernel is located on the right, a crucial point due to the noncommutative nature of the quaternion product. Eq. (\ref{eq:definitionQFT}) shows that a bivariate signal $x(t) \in \bbCi$ has a quaternion-valued spectral description $X(\nu) \in \bbH$. Moreover the QFT of $\bbCi$-valued signals exhibits the $\bmi$-Hermitian symmetry \cite{Bihan2014}
\begin{equation}\label{eq:ihermitian}
    X(-\nu) = \overline{X(\nu)}^{\bmi}.
\end{equation}
Eq. (\ref{eq:ihermitian}) illustrates that for bivariate signals negative frequencies carry no information additional to positive frequencies.
In \cite{Flamant2017timefrequency} we demonstrated that it permits to construct a direct bivariate counterpart of the usual analytic signal by canceling out negative frequencies of the spectrum.
This first tool called the \emph{quaternion embedding of a complex signal} allows identification of both instantaneous phase and polarization (\ie geometric) properties of narrow-band bivariate signals.
This approach can be extended to wideband signals using a a \emph{polarization spectrogram} based on a short-time QFT.
See \cite{Flamant2017timefrequency} for details.

For finite energy signals a generalized Parseval-Plancherel theorem gives yields two invariants:
\begin{align}
    \intinf \vert x(t) \vert^2 \dt &= \intinf \vert X(\nu)\vert^2\dnu,\label{eq:parseval_energy}\\
    \intinf \polarmodj{x(t)} \dt &= \intinf \polarmodj{X(\nu)}\dnu\label{eq:parseval_polar}.
\end{align}
Eq. (\ref{eq:parseval_energy}) is classical, \emph{energy} is conserved.
Eq. (\ref{eq:parseval_polar}) illustrates that an additional quadratic quantity of geometric nature is conserved. Importantly, the term $\polarmodj{X(\nu)} \in \mathrm{span}\lbrace \bmi, \bmj, \bmk\rbrace$ represents a vector in $\bbR^3$ which can be meaningfully interpreted in terms of polarization attributes \cite{flamant2017spectral,Flamant2017timefrequency}.

\subsection{Quaternion spectral density of bivariate signals} 
\label{sub:quaternion_spectral_density}
The QFT has two invariants (\ref{eq:parseval_energy}) and (\ref{eq:parseval_polar}). As a result for finite energy deterministic signals the quantities $\vert X(\nu)\vert^2$ and $\polarmodj{X(\nu)}$ summarize the second-order spectral properties of the bivariate signal $x(t)$. These quantities can be adequatly combined to form a \emph{quaternion energy spectral density}:
\begin{equation}\label{eq:energySpectralDensity}
  \Gamma_{xx}(\nu) = \vert X(\nu)\vert^2 + X(\nu)\bmj\overline{X(\nu)}.
\end{equation}

Many signals however are random and only of finite power, which makes the spectral density definition (\ref{eq:energySpectralDensity}) no longer applicable.
Fortunately thanks to a spectral representation theorem based on the QFT \cite{flamant2017spectral} one can extend the definition (\ref{eq:energySpectralDensity}) to define a quaternion power spectral density for stationary random bivariate signals.
In short the standard QFT $X(\nu)$ is replaced by the spectral increment $\mathrm{d}X(\nu)$: see Appendix \ref{app:spectralRep} for details.
Note however that for ease of notation we will make the slight abuse of writing $X(\nu)$ either when $x(t)$ is random, keeping in mind the correspondence described in Appendix \ref{app:spectralRep}.

The \emph{quaternion power spectral density} of a stationary random signal $x(t)$ reads:
\begin{equation}\label{eq:Gamma_xx_explicit_rwting}
  \Gamma_{xx}(\nu) = \underbrace{S_{0, x}(\nu)}_{\mathrm{scalar\:part}}+ \underbrace{\Phi_{x}(\nu)S_{0, x}(\nu)\bmmu_{x}(\nu)}_{\mathrm{vector\: part}}.
\end{equation}
The scalar part of $\Gamma_{xx}(\nu)$, $S_{0, x}(\nu) \geq 0$ is standard and gives the total\footnote{The term ``total'' refers to the fact that $S_{0, x}(\nu)$ contains power contributions from the \emph{unpolarized} and \emph{polarized} part, see \cite{flamant2017spectral}} power spectral distribution. The vector part of $\Gamma_{xx}(\nu)$ describes the \emph{polarization} properties of $x$ at every frequency.
They are summarized by two parameters: the \emph{polarization axis} $\bmmu_x(\nu)$, a pure unit quaternion, describes the polarization ellipse at this frequency.
The \emph{degree of polarization} $\Phi_x(\nu) \in [0, 1]$ quantifies the balance between
\emph{polarized} and \emph{unpolarized} parts at this frequency.
When $\Phi_{x}(\nu) = 0$ (resp. $ = 1$) the signal is \emph{unpolarized} (resp. \emph{fully polarized}) at $\nu$; else it is \emph{partially} polarized.

Figure \ref{fig:poincaresphere} depicts the Poincar\'e sphere of polarization states.
It allows a direct geometric interpretation of the vector part of the spectral density, \ie of polarization properties.
Normalizing in (\ref{eq:Gamma_xx_explicit_rwting}) the vector part of $\Gamma_{xx}(\nu)$ by the power distribution $S_{0, x}(\nu)$ gives the pure quaternion $\Phi_x(\nu)\bmmu_x(\nu)$.
Given any $\nu$ this quaternion identifies a vector of $\bbR^3$.
It is represented as a point on the surface of Poincar\'e sphere of radius $\Phi_{x}(\nu)$.
This point encoded by the pure unit quaternion $\bmmu_{x}(\nu)$ gives the \emph{polarization ellipse} of the signal at frequency $\nu$.
For instance, $\bmmu_x(\nu) = \bmi$ corresponds to counter-clockwise circular polarization, while $\bmmu_x(\nu) = -\bmj$ corresponds to vertical linear polarization.
Equivalently, $\bmmu_x(\nu)$ can be specified using spherical coordinates $(2\theta, 2\chi)$, giving respectively the orientation $\theta$ and ellipticity $\chi$ of the polarization ellipse; $\bmmu_x(\nu)$ can also be specified in Cartesian coordinates using normalized Stokes parameters,  see \emph{e.g} \cite{born2000principles} for details.
\emph{Orthogonal polarizations} correspond to antipodal points on the Poincar\'e sphere of radius $\Phi_x = 1$:
\emph{e.g.} clockwise and counter-clockwise circular are orthogonal polarizations.
While it may sound disturbing at first, two axes $\bmmu_x$ and $\bmmu_y$ correspond to orthogonal polarizations in the usual sense when they are anti-aligned $\ip{\bmmu_x, \bmmu_y} = -1$.

\begin{figure}
    \centering
    \includegraphics[width=0.4\textwidth]{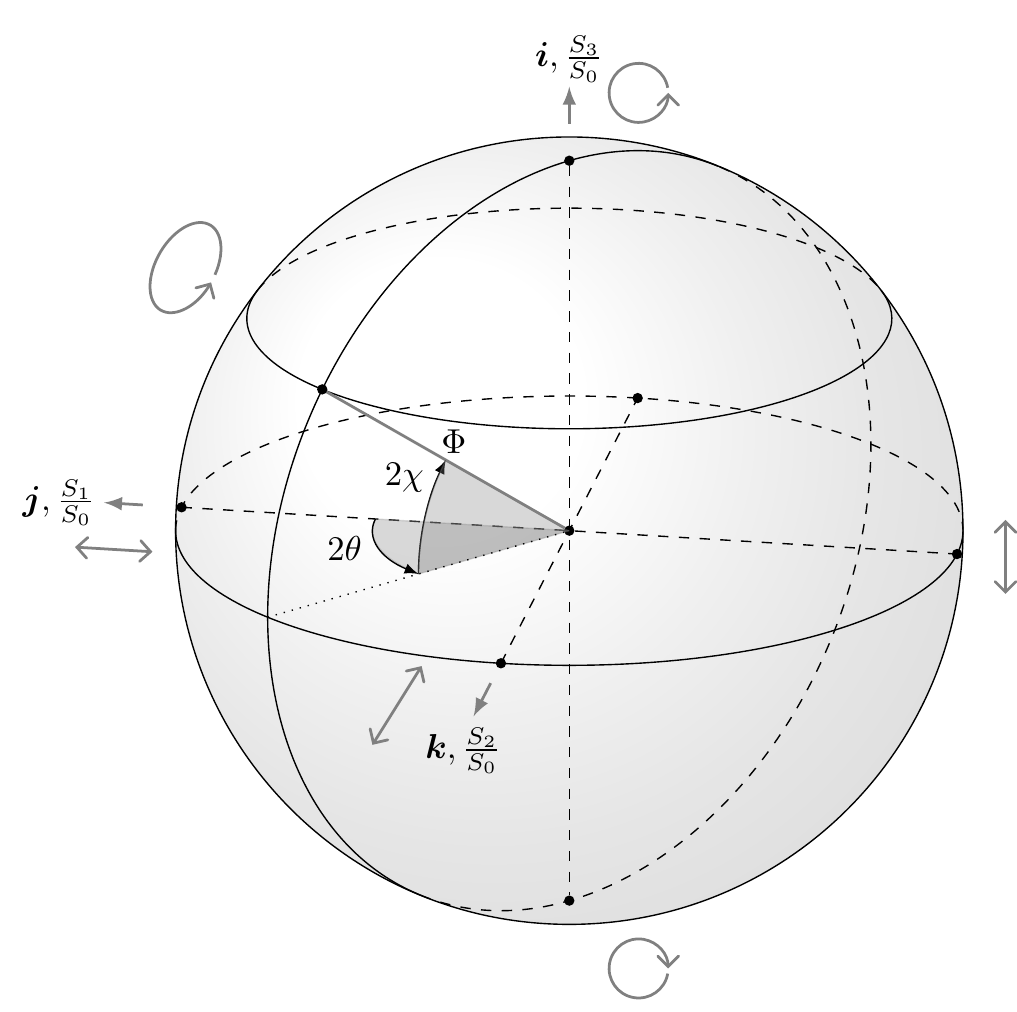}\caption{Poincar\'{e} sphere of polarization states. The vector part of $\Gamma_{xx}(\nu)$ (\ref{eq:Gamma_xx_explicit_rwting}) normalized by $S_{0, x}(\nu)$ identifies a vector in $\bbR^3$ which describes the polarization attributes of $x(t)$ at frequency $\nu$.
    Spherical coordinates $(2\theta, 2\chi)$ gives the orientation $\theta$ and ellipticity $\chi$. The radius $\Phi$ gives the degree of polarization.
    Cartesian coordinates give the normalized Stokes parameters, an alternative characterization of polarization properties \cite{born2000principles}.
}
    \label{fig:poincaresphere}
\end{figure}

\section{LTI filtering for bivariate signals}
\label{sec:filters_for_bivariate_signals}

  The purpose of this section is to write a complete and clean formulation of the theory of linear-time invariant (LTI) filtering for bivariate signals within the QFT framework.

  LTI filters can be classified into two categories: unitary filters and Hermitian filters.
  This decomposition originates from optics, where one usually separates \emph{birefringence} effects (unitary) from \emph{diattenuation} or \emph{dichroism} effects (Hermitian) \cite{J.J.Gil2007,J.J.Gil2016}.
  It is often implicitly assumed that one operates at a single frequency.
  In contrast we provide frequency-dependent expressions for unitary and Hermitian filters to deal with generic wideband bivariate signals.
  It must be pointed out that in general, in the time-domain there is no simple form involving a convolution for these filters.

  The quaternion representation offers a direct description of these filters in terms of \emph{birefringence} and \emph{diattenuation} parameters.
  Precisely, the use of quaternion algebra allows to write unitary and Hermitian filters in terms of eigenvectors and eigenvalues of their matrix representation.
  It explicitely uses the \emph{eigenpolarizations} of the filter, giving a natural way to identify the parameters of each filter.

  Section \ref{sub:matrix_representation} recalls that any LTI filter can be decomposed, at each frequency, into the combination of a unitary and a Hermitian transform.
  Lemmas \ref{prop:unitaryTransform} and \ref{prop:hermitianTransform} give quaternion representations of such transforms.
  Section \ref{sub:unitary_filters} and \ref{sub:hermitian_filters} study \emph{unitary} filters and \emph{Hermitian} filters, respectively.
  We emphasize physical and geometric interpretations of these two filters. See Appendix \ref{sec:app:linear_algebra_and_quaternion_equivalence} for technical details.

  \subsection{Matrix and quaternion representation}
  \label{sub:matrix_representation}
    In the following, time-domain (resp. frequency-domain) quantities are given in lowercase letters (resp. uppercase).
    Scalar quantities (in general, quaternion-valued) are denoted by standard case letters $x, X$.
    Vectors are denoted by bold straight letters $\ve{x}, \ve{X}$ and matrices are written as bold straight underlined letters $\mat{m}$, $\mat{M}$.
    Vector and matrices are always complex $\bbCj$-valued.

    A generic LTI filter is described by its matrix impulse response $\mat{m}(t) \in \matrixspace{\bbC}{2}_{\bmj}$ or by its Fourier Transform (FT) $\mat{M}(\nu) \in \matrixspace{\bbC}{2}_{\bmj}$.
    In the frequency domain the filtering relation between bivariate signals $\ve{x}$ and $\ve{y}$ reads:
    \begin{equation}\label{eq:filteringMatrixUsualFT}
    \ve{Y}(\nu) = \mat{M}(\nu)\ve{X}(\nu).
    \end{equation}
    For each $\nu$, Eq (\ref{eq:filteringMatrixUsualFT}) defines a linear relation between vectors $\ve{Y}(\nu)$ and $\ve{X}(\nu)$.
    For the rest of this section we fix $\nu$ and drop now this dependence.
    The \emph{polar decomposition} \cite{lancaster1985theory} of $\mat{M}$ is
    \begin{equation}\label{eq:polarDecompositionMatrix}
    \mat{M} = \mat{U}\:\mat{H},
    \end{equation}
    where $\mat{U}$ is unitary and $\mat{H}$ is Hermitian semi-definite positive, \ie $\mat{H}^* = \mat{H}$ and its eigenvalues are nonnegative.
    Geometrically (\ref{eq:polarDecompositionMatrix}) decomposes $\mat{M}$ as a stretch (Hermitian matrix $\mat{H}$) followed by a rotation (unitary matrix $\mat{U}$).
    %
    The polar decomposition (\ref{eq:polarDecompositionMatrix}) suggests to study separately two fundamental transforms, respectively \emph{unitary} and \emph{Hermitian} ones.
    Remarkably these two transforms have a direct interpretation in the quaternion representation. In particular parameters are directly related to eigenvectors and eigenvalues of each transform.

    Recall the equivalence between vector and quaternion representations:
    \begin{equation}\label{eq:spectralVectorQuaternionRep}
      \ve{X} = [X_1, X_2]^T \in \bbCj^2 \longleftrightarrow X = X_1 + \bmi X_2 \in \bbH, X_1, X_2 \in \bbCj.
    \end{equation}
    Lemma \ref{prop:unitaryTransform} gives the representation of unitary transforms in the quaternion domain.

    \begin{lemma}[Unitary transform]\label{prop:unitaryTransform}
      Let $\mat{U} \in \Utwo$. Then
      \begin{equation}\label{eq:prop_unitaryTransform}
          \ve{Y} = \mat{U}\ve{X} \Longleftrightarrow Y = e^{\bmmu\frac{\alpha}{2}}Xe^{\bmj\phi}
      \end{equation}
      where $\bmmu^2 = -1$, and $\alpha, \phi \in [0, 2\pi)$.
    \end{lemma}
    The proof is given in Appendix \ref{sub:app:unitary_transforms}.
    The parameter $\phi$ is the argument of $\det\mat{U}$. When $\phi = 0$, $\mat{U} \in \SUtwo$, \ie
    $\mat{U}$ is unitary with unit determinant, and (\ref{eq:prop_unitaryTransform}) highlights the well known \cite{altmann2005rotations} quaternion representation of special unitary matrices.
    The parameter $\bmmu$ gives the eigenvectors of $\mat{U}$, while $\alpha$ encodes its eigenvalues, see Appendix \ref{sub:app:unitary_transforms}.

    Lemma \ref{prop:hermitianTransform} gives the representation of Hermitian transforms in the quaternion domain.
    \begin{lemma}[Hermitian transform]\label{prop:hermitianTransform} Let $\mat{H} \in \matrixspace{\bbC}{2}_{\bmj}$ be Hermitian positive semi-definite. Then
    \begin{equation}\label{eq:prop_hermitian_transform}
        \ve{Y} = \mat{H}\ve{X} \Longleftrightarrow Y = K[X - \eta \bmmu X \bmj]
    \end{equation}
        where $\bmmu^2 = -1$, $K \in \bbR^+$ and $\eta \in [0, 1]$.
    \end{lemma}
    The proof is given in Appendix \ref{sub:app:hermitian_transforms}.
    The parameter $\bmmu$ encodes the eigenvectors of $\mat{H}$. Parameters $K$ and $\eta$ depend on, respectively, the sum and difference of eigenvalues, see Appendix \ref{sub:app:hermitian_transforms}.

    The quaternion representation allows a direct interpretation and control of each transform parameters.
    More importantly these key results enable efficient design of \emph{unitary} and \emph{Hermitian} filters, see Sections \ref{sub:unitary_filters} and \ref{sub:hermitian_filters} below.

  \subsection{Unitary filters}
  \label{sub:unitary_filters}
  A unitary filter performs a unitary transform for each frequency.
  Such filter only modifies the \emph{polarization axis} of the input signal: the total PSD and degree of polarization are not affected.
  It is defined by three frequency-dependent quantities: a \emph{birefringence axis} $\bmmu(\nu)$, a \emph{birefringence angle} $\alpha(\nu)$ and \emph{phase} $\phi(\nu)$. The parameter $\phi(\nu)$ is classical and quantifies the \emph{time delay} associated to each frequency.
  Quantities $\bmmu(\nu)$ and $\alpha(\nu)$ model \emph{birefringence} \cite{J.J.Gil2016, J.J.Gil2007}.
  This phenomenom is of fundamental importance in many areas \emph{e.g.} optical fiber transmission \cite{Gordon2000,Francia1998}.

  Proposition \ref{prop:unitary_filter} gives the unitary filtering relation for bivariate signals.
  Relations between corresponding quaternion spectral densities are given below, which permit further physical and geometric interpretations.
  \begin{proposition}[Unitary filter]\label{prop:unitary_filter}
      Let $x$ be the input and $y$ be the output of the unitary filter, with respective QFTs $X$ and $Y$.
      The filtering relation is
      \begin{equation}\label{eq:filtering_unitary}
         Y(\nu) = e^{\bmmu(\nu)\frac{\alpha(\nu)}{2}}X(\nu)e^{\bmj\phi(\nu)},
      \end{equation}
      with $\bmmu(-\nu) = \involi{\bmmu(\nu)}$, $\alpha(-\nu) = \alpha(\nu)$ and $\phi(-\nu) = -\phi(\nu)$.
      The spectral density of $y$ is
      \begin{equation}\label{eq:interference_unitary}
          \Gamma_{yy}(\nu) = e^{\bmmu(\nu)\frac{\alpha(\nu)}{2}}\Gamma_{xx}(\nu)e^{-\bmmu(\nu)\frac{\alpha(\nu)}{2}}
      \end{equation}
  \end{proposition}
  \begin{proof}[Sketch of proof]
    Eq. (\ref{eq:filtering_unitary}) is obtained directly from Lemma \ref{prop:unitaryTransform}.
    To obtain (\ref{eq:interference_unitary}) use the correspondence described in Appendix \ref{app:spectralRep}.
    Plugging (\ref{eq:filtering_unitary}) into the spectral density definition (\ref{appeq:quaternionSpectralDensity}) yields (\ref{eq:interference_unitary}).
  \end{proof}

  Symmetry conditions in (\ref{eq:filtering_unitary}) ensure that the $\bmi$-Hermitian symmetry (\ref{eq:ihermitian}) is satisfied for $Y(\nu)$ so that $y(t)$ is $\bbCi$-valued.
  Plugging (\ref{eq:Gamma_xx_explicit_rwting}) into (\ref{eq:interference_unitary}) yields
  \begin{equation}\label{eq:interference_unitary_rwting}
  \begin{split}
      \Gamma_{yy}(\nu) &=  e^{\bmmu(\nu)\frac{\alpha(\nu)}{2}} S_{0, x}(\nu)[1 + \Phi_{x}(\nu)\bmmu_{x}(\nu)]e^{-\bmmu(\nu)\frac{\alpha(\nu)}{2}} \\
       &= S_{0, x}(\nu) + \Phi_{x}(\nu)e^{\bmmu(\nu)\frac{\alpha(\nu)}{2}}\bmmu_x(\nu)e^{-\bmmu(\nu)\frac{\alpha(\nu)}{2}}.
  \end{split}
  \end{equation}
  Eqs. (\ref{eq:interference_unitary})--(\ref{eq:interference_unitary_rwting}) show that the unitary filter performs a \emph{geometric operation}: a 3D rotation of the spectral density $\Gamma_{xx}(\nu)$.
  Birefringence affects the output polarization axis $\bmmu_y(\nu)$, which is given by the rotation of the input polarization axis $\bmmu_x(\nu)$.
  Birefringence axis $\bmmu(\nu)$ and angle $\alpha(\nu)$ define this rotation. This geometrical operation can be visualized on the Poincar\'e sphere in Fig. \ref{fig:poincaresphere}.
  Eq. (\ref{eq:interference_unitary_rwting}) highlights that the total PSD and degree of polarization are rotation invariant: $S_{0,y}(\nu) = S_{0, x}(\nu)$ and $\Phi_y(\nu) = \Phi_x(\nu)$.
  The output polarization axis $\bmmu_y(\nu)$ is given by the rotation of angle $\alpha(\nu)$ of $\bmmu_x(\nu)$ around the axis $\bmmu(\nu)$.

  \textbf{Eigenpolarizations}. At a given $\nu$, unitary filters have two orthogonal eigenpolarizations. These are fully polarized spectral components $Z_\pm(\nu)$ with polarization axis is $\bmmu_{z_\pm}(\nu) = \pm \bmmu(\nu)$. As as result one gets
  \begin{equation}\label{eq:unitary_eigenpolarizations}
      e^{\bmmu(\nu)\frac{\alpha(\nu)}{2}}Z_{\pm}(\nu)e^{\bmj\phi(\nu)} = Z_{\pm}(\nu)e^{\bmj(\phi(\nu)\pm\alpha(\nu)/2)}.
  \end{equation}
  Eq. (\ref{eq:unitary_eigenpolarizations}) is another illustration of birefringence. It shows that unitary filters introduce a phase difference $\alpha(\nu)$ between the \emph{fast} eigenpolarization $Z_+(\nu)$ and \emph{slow} eigenpolarization $Z_-(\nu)$.

  Eigenpolarizations properties (\ref{eq:unitary_eigenpolarizations}) give a simple way to identify the parameters of the filter. The approach is analogous to what is done in experimental optics \cite{J.J.Gil2016}.
  Working with monochromatic signals of increasing frequency, one can adjust input polarization axis such that the output polarization axis are the same.
  It gives immediatly the birefringence axis $\bmmu(\nu)$. Measuring phase delays with respect to fast and slow eigenpolarizations then permits using (\ref{eq:unitary_eigenpolarizations}) to identify birefringence angle $\alpha(\nu)$ and phase $\phi(\nu)$.

  \subsection{Hermitian filters}
  \label{sub:hermitian_filters}
  A Hermitian filter performs a Hermitian transform at each frequency.
  This second type of filter acts on both power and polarization properties of the input signal.
  Three frequency-dependent quantities are necessary to define a Hermitian filter: the \emph{homogeneous gain} $K(\nu) \geq 0$ and two quantities related to \emph{diattenuation}: the \emph{polarizing power} $\eta(\nu)$ and the \emph{diattenuation axis} $\bmmu(\nu)$.
  When $\eta(\nu) = 0$, $K(\nu)$ has a classical interpretation as the gain of the filter.
  When $\eta(\nu) \neq 0$ the gain of the filter depends on the projection of the polarization axis $\bmmu_x(\nu)$ onto the diattenuation axis $\bmmu(\nu)$.
  In particular eigenpolarizations, which are spectral components with polarization axis $\pm \bmmu(\nu)$ correspond to maximum and minimum gain values.

  Proposition \ref{prop:hermitian_filter} gives the Hermitian filtering relation for bivariate signals.
  Relations between input and output spectral densities are presented.
  The use of (\ref{eq:Gamma_xx_explicit_rwting}) yields an explicit rewriting of
  $\Gamma_{yy}(\nu)$ in terms of input polarization properties.
  \begin{proposition}[Hermitian filter]\label{prop:hermitian_filter}
      Let $x$ be the input and $y$ be the output of the Hermitian filter, with respective QFTs $X$ and $Y$.
      The filtering relation is
      \begin{equation}\label{eq:hermitian_filtering}
        Y(\nu) = K(\nu)[X(\nu) - \eta(\nu) \bmmu(\nu) X(\nu)\bmj]
      \end{equation}
      with $K(-\nu) = K(\nu)$, $\eta(-\nu) = \eta(\nu)$ and $\bmmu(-\nu) = \conji{\bmmu(\nu)}$.
      Using (\ref{eq:Gamma_xx_explicit_rwting}), the spectral density of $y$ is then given by (dropping $\nu$ dependence for convenience)
      \begin{align}
          \mathcal{S}\left(\Gamma_{yy}\right) &=S_{0, x} K^2 \left[ 1 + \eta^2 +2\eta \Phi_x\ip{\bmmu, \bmmu_x} \right]\label{eq:hermitianFilterSOy} \\
          \mathcal{V}\left(\Gamma_{yy}\right)&= S_{0, x} K^2\left[ 2\eta\bmmu + \Phi_x[\bmmu_x - \eta^2\bmmu\bmmu_x\bmmu
          \right]\label{eq:hermitianFilterVectorial}
      \end{align}
      where $\ip{\bmmu_1, \bmmu_2} = \mathcal{S}(\bmmu_1\overline{\bmmu_2})$ is the usual inner product of $\bbR^3$.
  \end{proposition}
  \begin{proof}[Sketch of proof]
    Eq. (\ref{eq:hermitian_filtering}) is obtained directly from Lemma \ref{prop:hermitianTransform}.
    To obtain (\ref{eq:hermitianFilterSOy})-(\ref{eq:hermitianFilterVectorial}) use the correspondence described in Appendix \ref{app:spectralRep}.
    Plugging (\ref{eq:hermitian_filtering}) into the spectral density definition (\ref{appeq:quaternionSpectralDensity}) with the use of (\ref{eq:Gamma_xx_explicit_rwting}) yields (\ref{eq:interference_unitary}).
  \end{proof}
  Symmetry conditions in (\ref{eq:hermitian_filtering}) ensure that the $\bmi$-Hermitian symmetry (\ref{eq:ihermitian}) is satisfied for $Y(\nu)$ so that $y(t)$ is $\bbCi$-valued.
  In the sequel, we work at a fixed frequency $\nu$.
  Explicit dependence in $\nu$ is dropped to avoid notational clutter.

  \textbf{Gain}. The \emph{power gain} $G$ of the filter is defined by
  \begin{equation}
    G = \frac{\mathcal{S}\left(\Gamma_{yy}\right)}{\mathcal{S}\left(\Gamma_{xx}\right)} = \frac{S_{0,y}}{S_{0, x}}
  \end{equation}
  Using Eq. (\ref{eq:hermitianFilterSOy}) this gain becomes
  \begin{equation}\label{eq:gain_hermitian}
      G =  K^2\left[1 + \eta^2 +2\eta \:\Phi_x\ip{\bmmu, \bmmu_x} \right].
  \end{equation}
  When $\eta = 0$ the power gain reduces to its usual expression $G = K^2$.
  When $\eta \neq 0$, the gain depends on $K$ and $\eta$ but most importantly, on the alignment $\ip{\bmmu, \bmmu_x}$ between diattenuation and input polarization axes.

  \textbf{Eigenpolarizations}. Hermitian filters have two orthogonal eigenpolarizations.
  These are fully polarized spectral components $Z_\pm$ with polarization axis
  $\bmmu_{z_\pm} = \pm \bmmu$. From (\ref{eq:hermitian_filtering}) one has
  \begin{equation}\label{eq:hermitian_eigenpolarizations}
     K[Z_\pm - \eta \bmmu Z_\pm \bmj] = K[1\pm\eta] Z_\pm.
  \end{equation}
  Eq. (\ref{eq:hermitian_eigenpolarizations}) characterizes \emph{diattenuation} \cite{J.J.Gil2016, J.J.Gil2007}.
  Orthogonal eigenpolarizations have different gains; the \emph{polarizing power} $\eta$ controls the gap between respective gain values.

  As with the unitary filter, eigenpolarization properties (\ref{eq:hermitian_eigenpolarizations}) give a natural way to identify filter parameters.
  Note first that eigenpolarizations correspond directly to maximum and minimum values of the gain $G$ (\ref{eq:gain_hermitian}).
  Thus, finding the maximum and minimum value of the gain by changing the input polarization allows to identify directly parameters $K$, $\eta$ and $\bmmu$.
  Let $G_{\text{max}}$ and $G_{\text{min}}$ denote the maximal/minimal gain values, one has
  \begin{equation}
    \frac{2\eta}{1+\eta^2} = \frac{G_{\text{max}} - G_{\text{min}}}{G_{\text{max}}+ G_{\text{min}}} \text{ and } K^2 = \frac{G_{\text{max}} - G_{\text{min}}}{4\eta}.
  \end{equation}
  Repeating the operation for a wide range of frequencies completes the characterization procedure.

  \textbf{Identification using unpolarized WGN}.
  The spectral density of the response of the Hermitian filter to an unpolarized white Gaussian noise input provides a simple and practical way to identify its parameters.
  The input unpolarized WGN noise $w(t)$ has constant spectral density $\Gamma_{ww}(\nu) = \sigma_0^2 \geq 0$, with $\sigma_0^2$ the noise variance.
  It is unpolarized for every frequency since $\Phi_w(\nu) = 0$.
  Then the output $y(t)$ has spectral density
  \begin{equation}\label{eq:unpolarized_input_hermitian}
    \Gamma_{yy}(\nu) = \sigma_0^2K^2(\nu)[1 + \eta^2(\nu)  + 2\eta(\nu) \bmmu(\nu)].
  \end{equation}
  Filter parameters $\eta(\nu)$ and $\bmmu(\nu)$ completely define the output polarization state.
  Identifying (\ref{eq:Gamma_xx_explicit_rwting}) for $\Gamma_{yy}$ with (\ref{eq:unpolarized_input_hermitian}) yields the filter parameters:
  \begin{equation}\label{eq:filterParametersFromUnpolarized}
  \begin{cases}
      \eta(\nu) &= \displaystyle\frac{1-\sqrt{1-\Phi_y^2(\nu)}}{\Phi_y(\nu)} \quad (\Phi_y(\nu)\neq 0)\\
      K^2(\nu) &= \displaystyle \frac{S_{0, y}(\nu)}{\sigma_0^2(1+\eta^2(\nu))}\\
      \bmmu(\nu) &= \displaystyle\bmmu_y(\nu)
  \end{cases}
  \end{equation}
  and $\eta(\nu) = 0$ when $\Phi_y(\nu) = 0$.

  This result is fundamental.
  In the bivariate case, unpolarized white noise plays the role of white noise in the univariate case.
  It permits a direct identification of the parameters of the Hermitian filter.
  Moreover any bivariate signal with arbitrary spectral density $\Gamma_{yy}$ can be obtained as a Hermitian filtered version of unpolarized white noise.
  Section \ref{sub:spectral_synthesis} exploits the latter property to simulate stationary bivariate signals via spectral synthesis.

  \textbf{Examples}.
  Hermitian filters are characterized by non-trivial interactions between input polarization properties and filter parameters.
  Two particular cases illustrate how far the proposed approach is rich and interpretable.
  Frequency dependence is omitted in what follows.

  \emph{Null polarizing power $\eta = 0$}. One has $Y = KX$ and $\Gamma_{yy} = K^2 \Gamma_{xx}$.
  The output is a purely amplified/attenuated version of the input signal.
  Polarization properties are not modified.

  \emph{Maximal polarizing power $\eta = 1$}. The Hermitian filter is called a \emph{polarizer} since the output polarization properties do not depend on the input polarization properties.
  Geometrically, starting from (\ref{eq:hermitianFilterVectorial}) the term $\bmmu_x - \bmmu\bmmu_x\bmmu$ corresponds to the projection of $\bmmu_x$ onto $\bmmu$, up to a factor 2: the filter performs a projection onto the diattenuation axis $\bmmu$.
  The output polarization axis is $\bmmu_y = \bmmu$; the output is totally polarized $\Phi_y = 1$.
  The gain $G$ quantifies how `close' $\bmmu_x$ is to $\bmmu$:
  \begin{equation}
    G = 2 S_{0, x} K^2[1 + \Phi_x \ip{\bmmu_x, \bmmu}]
  \end{equation}
  In particular, for eigenpolarizations $Z_\pm$:
  \begin{equation}
    Y_+ = 2 K Z_+ \text{ and } Y_- = 0
  \end{equation}
  meaning that when the input polarization axis is $\bmmu_x = -\bmmu$ (orthogonal polarization) and totally polarized, the output cancels out.
  It illustrates how the alignment between input polarization and diattenuation axes affects the gain of the filter.

\section{Applications}
\label{sec:applications}
  \begin{figure*}
    \includegraphics[width=\textwidth]{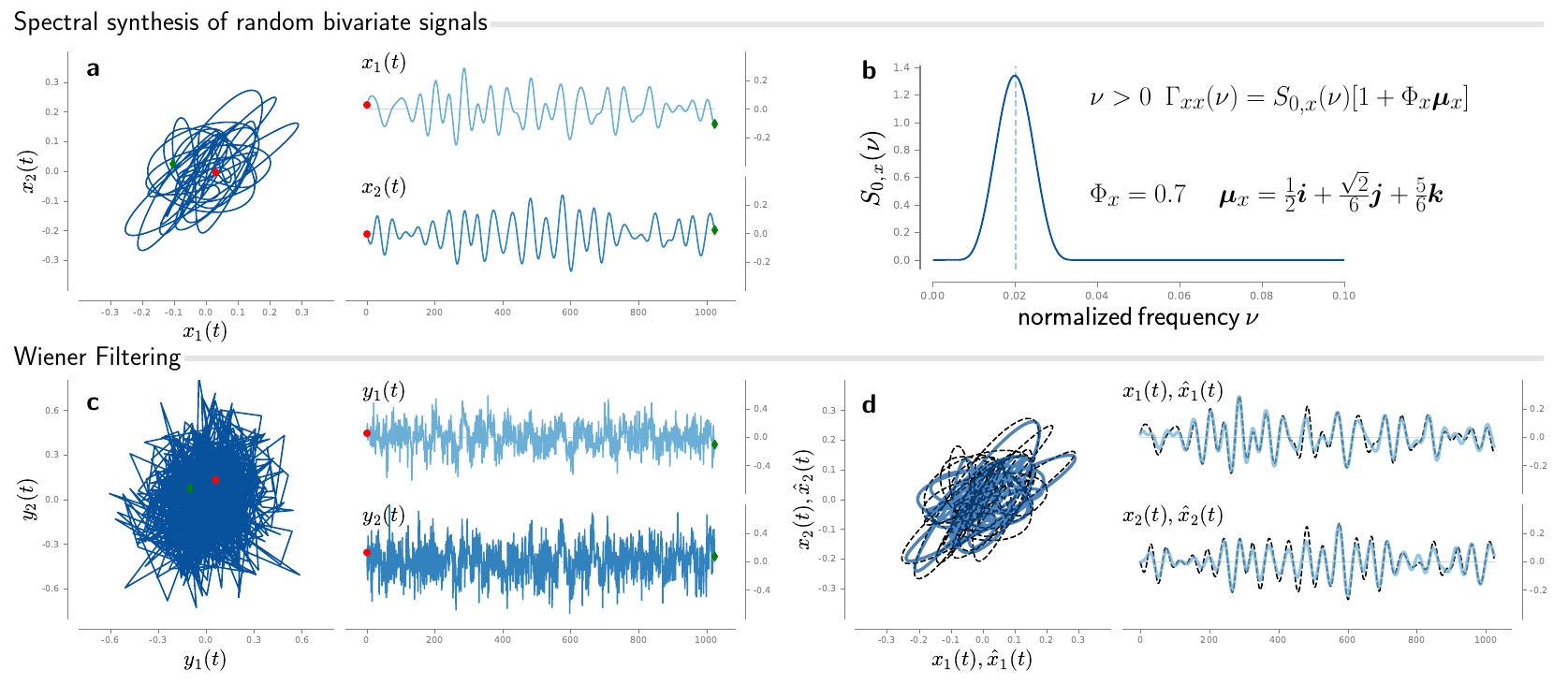}
    \caption{Numerical simulations illustrating the novel tools introduced in this paper.
    \textbf{\textsf{a}}\: Partially elliptically polarized narrow-band signal obtained using the spectral synthesis method of Sec. \ref{sub:spectral_synthesis}. This reference signal is used in all subsequent simulations.
    \textbf{\textsf{b}}\: Power spectral distribution and parameters used in \textbf{\textsf{a}}.
    \textbf{\textsf{c}} Reference signal in partially ($\Phi_w = 0.4)$ vertically polarized white noise with $\mathrm{SNR} = -5$ dB.
    \textbf{\textsf{d}}
    Output of the Wiener denoising filter described in Sec. \ref{sub:wiener_denoising}. Dashed lines indicate the original signal of Fig. 2a.}
    \label{fig:experimental}
  \end{figure*}
  \subsection{Spectral synthesis}
  \label{sub:spectral_synthesis}

  We propose a new simulation method for Gaussian stationary random bivariate signals based on the filtering of a bivariate white Gaussian noise.
  Eq. (\ref{eq:unpolarized_input_hermitian}) shows that any bivariate signal with arbitrary spectral density can be obtained by Hermitian filtering of unpolarized white noise.
  This result allows to generalize a well-known approximate simulation algorithm \cite{Percival1992} to the case of bivariate random signals.

  Let $\Gamma_0(\nu) = S_0(\nu)[1 + \Phi_0(\nu)\bmmu_0(\nu)]$ denote the spectral density of the \emph{target} signal to sample from.
  Let $w(t)$ be an unpolarized white noise: its spectral density is constant $\Gamma_{ww}(\nu) = \sigma_0^2 \in \bbR^+$.
  Let $x(t)$ be the result of Hermitian filtering of $w(t)$.
  Adapting notations from (\ref{eq:unpolarized_input_hermitian}) one gets
  \begin{equation}\label{eq:sub:spectral_synth_output_unidentified}
      \Gamma_{xx}(\nu) = \sigma_0^2K^2(\nu)[1 + \eta^2(\nu)]\left[1 + \frac{2\eta(\nu)}{1 + \eta^2(\nu)}\bmmu(\nu)\right].
  \end{equation}
  Remark that (\ref{eq:sub:spectral_synth_output_unidentified}) is of the form (\ref{eq:Gamma_xx_explicit_rwting}).
  Identifying filter parameters to match the target spectral density $\Gamma_0(\nu)$
  yields the same expressions as in (\ref{eq:filterParametersFromUnpolarized}).

  In practice one wants to generate a discrete, $N$-length realization of the signal $x(t)$.
  One starts by generating an i.i.d unpolarized white noise sequence of length $M \geq N$ (see Appendix \ref{app:whitenoise}).
  Filtering this sequence thanks to discrete implementation of (\ref{eq:hermitian_filtering}) and keeping the first $N$ samples gives a discretized realization of the signal $x(t)$.
  As in the univariate setting \cite{Percival1992}, the quality of the simulation is increasing with $M$.

  Figure \ref{fig:experimental}a depicts a realization of a narrow-band stationary random bivariate signal with constant polarization properties.
  The simulation is of length $N = 1024$ and was obtained using a $M = 10N$ length unpolarized white noise sequence.
  The signal is partially polarized $\Phi_x = 0.7$ and exhibits elliptical polarization axis. The power is distributed in a Gaussian-shaped fashion around normalized frequency $\nu_0 = 0.02$, see Figure \ref{fig:experimental}b for details.
  Note that the instantaneous polarization state evolves with time.
  This is a feature of partial polarization for quasi-monochromatic signals with constant polarization axis.

  \subsection{Wiener denoising}
  \label{sub:wiener_denoising}

  Wiener filtering is an ubiquitous tool in signal processing.
  We show that the Wiener filter for bivariate signals has a convenient quaternion representation.
  It allows meaningful physical interpretations and a direct parametrization in terms of polarization parameters.
  We restrict our analysis to the denoising case.
  Our goal is to estimate a signal of interest $x(t)$ from which we have \emph{measurements} $y(t)$ of the form
  \begin{equation}\label{eq:denoising_quaternion_problem}
    y(t) = x(t) + w(t)
  \end{equation}
  where $w(t)$ is bivariate noise, independent from $x(t)$.
  All signals are assumed to be zero-mean, second-order stationary with known spectral densities.
  The Wiener filter solves the minimum-mean-square-error (MMSE) problem
  \begin{equation}\label{eq:wienerProblemquaternion}
    \mathrm{min}\: \Expe{\vert \hat{x}(t) - x(t)\vert^2}
  \end{equation}
  where $\hat{x}(t)$ is obtained by linear filtering of $y(t)$.
  Intuitively when searching for a polarized deterministic signal $x(t)$ in unpolarized noise $w(t)$, the Wiener filter should behave like a \emph{polarizer}.
  It means that every spectral component of $y$ is projected along the polarization axis $\bmmu_x(\nu)$.
  Fortunately, this intuition is proven right by the generic expression of the Wiener filter.

  Frequency dependence is omitted for convenience.
  The Wiener denoising filter is a Hermitian filter (see Appendix \ref{app:wienerFilter} for calculations):
  \begin{equation}\label{eq:wienerFilterquaternionGeneric}
    \hat{X} = \frac{S_{0, x}\left(1 - \Phi_x\Phi_y\ip{\bmmu_x,\bmmu_y} \right)}{S_{0, y}[1-\Phi_y^2]}\left[ Y -  \frac{\Phi_x\bmmu_x-\Phi_y\bmmu_y }{1 - \Phi_x\Phi_y\ip{\bmmu_x,\bmmu_y}}Y\bmj\right].
  \end{equation}

  Quantities $K(\nu), \bmmu(\nu), \eta(\nu)$ of Proposition \ref{prop:hermitian_filter} can be readily identified from (\ref{eq:wienerFilterquaternionGeneric}).
  Note the use of the explicit form (\ref{eq:Gamma_xx_explicit_rwting}) of  $\Gamma_{yy}(\nu) = \Gamma_{xx}(\nu) + \Gamma_{ww}(\nu)$ to simplify notations.

  In many situations the noise $w(t)$ can be assumed unpolarized for every frequency.
  Then $\Gamma_{ww}(\nu) = \sigma^2(\nu) \in \bbR^+$ and
  \begin{equation}
    \Gamma_{yy}(\nu) = \underbrace{S_{0, x}(\nu) +\sigma^2(\nu)}_{S_{0, y}(\nu)} + \underbrace{S_{0, x}(\nu)\Phi_x(\nu)\bmmu_x(\nu)}_{S_{0, y}(\nu)\Phi_y(\nu)\bmmu_y(\nu)}
  \end{equation}
  The polarization axis is not affected by the noise: $\bmmu_y(\nu) = \bmmu_x(\nu)$ for all $\nu$.
  We introduce $\alpha = S_{0, x}/\sigma^2$, the frequency-domain signal-to-noise ratio (SNR).
  The degree of polarization is $\Phi_y(\nu) = \alpha(\nu)\Phi_x(\nu)/(1+\alpha(\nu))$.
  The Wiener filter (\ref{eq:wienerFilterquaternionGeneric}) then simplifies to
  \begin{equation}\label{eq:wienerFilterquaternionSimplecase}
    \hat{X} = \frac{\alpha+\alpha^2[1-\Phi_x^2]}{1+2\alpha + \alpha^2[1-\Phi_x^2]}
    \left[Y -\frac{\Phi_x}{1+\alpha[1-\Phi_x^2]}\bmmu_x Y\bmj\right].
  \end{equation}
  The diattenuation axis of the filter is the polarization axis of the target $\bmmu_x$.
  Homogeneous gain and polarizing power depend on the target degree of polarization $\Phi_x$ and frequency-domain SNR $\alpha$.
  In particular, when $x$ is deterministic (hence totally polarized at all frequencies) then the Wiener filter reduces to
  \begin{equation}\label{eq:wienerFilterquaternionPolarizer}
    \hat{X}(\nu) = \frac{S_{0, x}(\nu)}{2 S_{0, x}(\nu) + \sigma^2(\nu)}\left[Y(\nu) -\bmmu_x(\nu)Y(\nu)\bmj\right].
  \end{equation}
  Eq. (\ref{eq:wienerFilterquaternionPolarizer}) defines a polarizer and validates our initial intuition. Each spectral component of $y$ is projected along the polarization axis $\bmmu_x(\nu)$.

  The MMSE is $\epsilon_\text{opt} = \Expe{\vert \hat{x}(t)-x(t)\vert^2}$ with $\hat{x}(t)$ given by (\ref{eq:wienerFilterquaternionGeneric}).
  The MMSE can be rewritten as a frequency domain integral (see Appendix \ref{app:wienerFilter})
  \begin{equation}\label{eq:errorWienerGeneric}
    \epsilon_{\text{opt}} = \int_{-\infty}^\infty \epsilon_{\text{opt}}(\nu)\mathrm{d}\nu
  \end{equation}
  where $\epsilon_{\text{opt}}(\nu)$ is:
  \begin{align}
    \epsilon_{\text{opt}}(\nu) &= S_{0, x}\left(1 - \frac{S_{0, x}}{S_{0, y}}\frac{1+\Phi_x^2-2\Phi_x\Phi_y\ip{\bmmu_x, \bmmu_y}}{1- \Phi_y^2}\right)\label{eq:optimalErrorSigObs}\\
    &\hspace{-2em}=  S_{0, x}\frac{1 - \Phi_w^2 + \alpha[1-\Phi_x^2]}{1 - \Phi_w^2 + \alpha^2[1-\Phi_x^2] + 2\alpha[1-\Phi_x\Phi_w\ip{\bmmu_x, \bmmu_w}]}.\label{eq:optimalErrorSigNoise}
  \end{align}
  Eqs (\ref{eq:optimalErrorSigObs})-(\ref{eq:optimalErrorSigNoise}) illustrate the dependence of the optimal error in terms of polarization properties of the signal $x$, observation $y$ or noise $w$.
  Fixing all parameters excepted $\ip{\bmmu_x, \bmmu_w}$ in (\ref{eq:optimalErrorSigNoise}), the optimal error is minimum when signal and noise exhibit orthogonal polarizations, \ie when their polarization axes are anti-aligned $\ip{\bmmu_x, \bmmu_w} = -1$.
  The error is maximum when signal and noise have same polarization $\ip{\bmmu_x, \bmmu_w} = 1$.
  Given $\alpha$, asymmetry between minimum and maximum values is accentuated for strongly polarized signal and noise ($\Phi_x, \Phi_w \simeq 1$).
  For $\alpha \gg 1$ (\ref{eq:optimalErrorSigNoise}) becomes $\epsilon_{\text{opt}}(\nu) \simeq S_{0, x}(\nu)/\alpha(\nu)$,
  while for $\alpha << 1$ one gets $\epsilon_{\text{opt}}(\nu) \simeq S_{0, x}(\nu)$, as expected.

  We conclude by a numerical example of Wiener filter denoising.
  The signal $x(t)$ is taken as the synthetized signal of Fig. \ref{fig:experimental}a.
  It is a partially elliptically polarized narrow-band signal.
  Spectral density parameters are given in Fig \ref{fig:experimental}b.
  Measurements $y(t)$ are obtained using (\ref{eq:denoising_quaternion_problem}) with $w(t)$ a partially vertically polarized white Gaussian noise, see Appendix \ref{app:whitenoise} for details.
  Its spectral density is $\Gamma_{ww}(\nu) = \sigma^2(1 - 0.4\bmj)$.
  Noise variance is adjusted so that $\text{SNR} = -5$ dB.

  Figure \ref{fig:experimental}c depicts the measurements $y(t)$.
  Clearly, noise level is larger on the vertical axis on account of the partial vertical polarization of $w(t)$.
  Figure \ref{fig:experimental}d shows the output of the Wiener filter.
  The reconstruction SNR is $10\log_{10}(\Vert x(t)\Vert_2^2/\Vert \hat{x}(t)-x(t)\Vert_2^2) = 9.92 $ dB, where $\Vert \cdot\Vert_2$ is the standard 2-norm.
  It illustrates the good performances in recovering the original signal $x(t)$.

  \subsection{Some decompositions of stationary bivariate signals}
  \label{sub:decomposition_bivariate_signals}
  \begin{table*}
    \caption{Different decompositions obtained by changing the homogeneous gain $K(\nu)$. }\label{table:decompositionBivSig}
    \centering
    \resizebox{.9\textwidth}{!}{
      \begin{tabularx}{1.1\textwidth}{ccccl}
      \toprule
      &$K(\nu)$ & $\Gamma_{x_a, x_a}(\nu)$ & $\Gamma_{x_b, x_b}(\nu)$ & {\textbf{correlation}}\\
      \cmidrule(l){2-4}\cmidrule(l){5-5}
      \textsf{\textbf{(i)}}&$\displaystyle \sqrt{\dfrac{\Phi_x(\nu)}{2(1+\Phi_x(\nu))}}$ & $S_{0, x}(\nu)\Phi_x(\nu)[1+\bmmu_x(\nu)]$ &
      \begin{tabular}[t]{ll}
        & $\kappa(\nu)S_{0,x}(\nu)\left[1 - \Phi(\nu)\bmmu_x(\nu)\right]$\vspace{0.5em}\\

        with & $\kappa(\nu) =\left(1+\Phi_x(\nu)-2(\Phi_x(\nu) + 1)K(\nu)\right)$\\[.5em]
        &$\displaystyle\Phi(\nu)  = \frac{1-2\Phi_x(\nu)+2[\Phi_x(\nu) + 1]K(\nu)}{1+\Phi_x(\nu)-2[\Phi_x(\nu) + 1]K(\nu)}$
      \end{tabular}
      & correlated \\[.5em]
      \midrule \\
      \textsf{\textbf{(ii)}}&$\displaystyle 1 - \frac{\Phi_x(\nu)}{\Phi_x(\nu) + 1-\sqrt{1-\Phi_x(\nu)^2}}$
      & $2S_{0,x}(\nu)K^2(\nu)[1+\Phi_x(\nu)][1+\bmmu_x(\nu)]$ & $S_{0,x}(\nu)[1-\Phi_x(\nu)]$ &correlated \\[2em]
      \midrule \\
       \textsf{\textbf{(iii)}} &$\displaystyle \frac{1}{2}$& $\frac{S_{0, x}(\nu)}{2}[1+\Phi_x(\nu)][1+\bmmu_x(\nu)]$&
      $\frac{S_{0, x}(\nu)}{2}[1-\Phi_x(\nu)][1-\bmmu_x(\nu)]$& uncorrelated \\[1em]
      \bottomrule
    \end{tabularx}
    }
  \end{table*}
  It is known \cite{flamant2017spectral,born2000principles} that the spectral density of a bivariate signal $x(t)$ can be uniquely decomposed as the sum of unpolarized and totally polarized spectral densities:
  \begin{align}
      \Gamma_{xx}(\nu) &= [1-\Phi_x(\nu)]S_{0, x}(\nu) + \Phi_x(\nu)S_{0, x}(\nu)[1 + \bmmu_x(\nu)]\nonumber\\
      &= \Gamma_{xx}^U(\nu) + \Gamma_{xx}^P(\nu),\label{eq:UPDecompositionSpectralDensities}
  \end{align}
  where superscripts $U$ and $P$ stand respectively for \emph{unpolarized} and \emph{polarized} parts.
  The decomposition (\ref{eq:UPDecompositionSpectralDensities}) motivates the search for decompositions of the bivariate signal $x(t)$ into two parts $x_a(t)$ and $x_b(t)$ such that
  \begin{equation}\label{eq:decompositionAB}
    x(t) = x_a(t) + x_b(t).
  \end{equation}
  Comparing (\ref{eq:decompositionAB}) with (\ref{eq:UPDecompositionSpectralDensities}), we search a linear filter such that $x_a(t)$ is fully polarized along $\bmmu_x(\nu)$ for every frequency.
  Additionaly the two parts should satisfy: (i) $x_a(t)$ has spectral density $\Gamma_{xx}^P(\nu)$;
  (ii) $x_b(t)$ is unpolarized for every frequency, with spectral density $\Gamma_{xx}^U(\nu)$;
  (iii) $x_a(t)$ and $x_b(t)$ are uncorrelated.
  Unfortunately no such linear filter exists.
  Each requirement corresponds to a distinct filter: only one requirement at a time can be met.

  Since unitary filters do not affect the degree of polarization or are not able to decorrelate two signals, it is necessary to use a Hermitian filter.
  Moreover since we search for $x_a(t)$ fully polarized along $\bmmu_x(\nu)$, one has to use a \emph{polarizer} along the polarization axis of $x(t)$:
  \begin{align}
    X_a(\nu) &= K(\nu)\left(X(\nu)-\bmmu_x(\nu)X(\nu)\bmj\right),\label{eq:filteringDecomposition1}\\
    X_b(\nu) &= X(\nu) - X_a(\nu)\nonumber \\ &\hspace{-1em}=\left(1-K(\nu)\right)\left(X(\nu)+\frac{K(\nu)}{1-K(\nu)}\bmmu_x(\nu)X(\nu)\bmj\right).\label{eq:filteringDecomposition2}
  \end{align}
  The second component $x_b(t)$ is such that (\ref{eq:decompositionAB}) holds.
  Note that in (\ref{eq:filteringDecomposition1})-(\ref{eq:filteringDecomposition2})
  the gain $K(\nu)$ is not fixed.
  Requirements (i), (ii) or (iii) correspond to distinct values of this gain.
  Stated differently, $K(\nu)$ rules the nature of the decomposition (\ref{eq:decompositionAB}).

  Table \ref{table:decompositionBivSig} summarizes expressions of the gain and spectral densities of $x_a(t)$ and $x_b(t)$ for requirements (i), (ii) and (iii).
  In addition correlation properties of the two components are given.
  To meet (i) the gain $K(\nu)$ is adjusted thanks to (\ref{eq:hermitianFilterSOy}) such that $\Gamma_{x_a, x_a}(\nu) = \Gamma_{xx}^P(\nu)$.
  However $x_b(t)$ is partially polarized and components are correlated.
  For (ii) starting from (\ref{eq:filteringDecomposition2}) and using (\ref{eq:hermitianFilterVectorial}) with $\bmmu(\nu) = -\bmmu_x(\nu)$ one computes the vector part of $\Gamma_{x_b, x_b}(\nu)$.
  Then the gain $K(\nu)$ is obtained by imposing $\Phi_b(\nu) = 0$ for every $\nu$.
  Fortunately the corresponding expression for $K(\nu)$ yields $\Gamma_{x_b, x_b}(\nu) = \Gamma_{xx}^U(\nu)$.
  The first component $x_a(t)$ is fully polarized like $x(t)$, but has weaker intensity than that of $\Gamma_{xx}^P(\nu)$. Components are also correlated.
  Finally (iii) is fulfilled by enforcing decorrelation between $x_a(t)$ and $x_b(t)$.
  See Appendix \ref{app:spectralRep} for technical details.
  Importantly $x_a(t)$ and $x_b(t)$ are both fully polarized with orthogonal polarization axes. Respective intensities are controlled by the degree of polarization $\Phi_x(\nu)$.
  Fig. \ref{fig:decompositions} illustrate decompositions (ii) and (iii) on the synthetized signal of Fig. \ref{fig:experimental}a. Decomposition (i) is not presented as it is similar to (iii), excepted that $x_b(t)$ is only (strongly) partially polarized.

  Taking another polarization axis in (\ref{eq:filteringDecomposition1})-(\ref{eq:filteringDecomposition2}) will not enable satisfying requirements (i)-(ii)-(iii).
  Indeed the filter corresponding to (ii) and defined in Table  \ref{table:decompositionBivSig} is the unique \emph{depolarizer} of $x(t)$, \ie the only filter that outputs an unpolarized signal from a partially polarized input ($\Phi_x < 1$).
  Moreover the unique linear filter producing decorrelated signals for $x_a(t)$ and $x_b(t)$ is the one defined by (iii) in Table \ref{table:decompositionBivSig}.

  This discussion answers an important and natural question.
  Since the decomposition (\ref{eq:UPDecompositionSpectralDensities}) holds, is it possible to decompose by linear filtering any bivariate signal into uncorrelated unpolarized and polarized components?
  Unfortunately the answer is negative.
  However, this hypothetical decomposition can still be used as a synthesis tool, as already shown \cite{flamant2017spectral}.
  Moreover in practical situations where such a decomposition may be needed, one can choose the appropriate filter according to the desired requirement (i), (ii) or (iii).

  \begin{figure*}
    \centering
    \includegraphics[width=\textwidth]{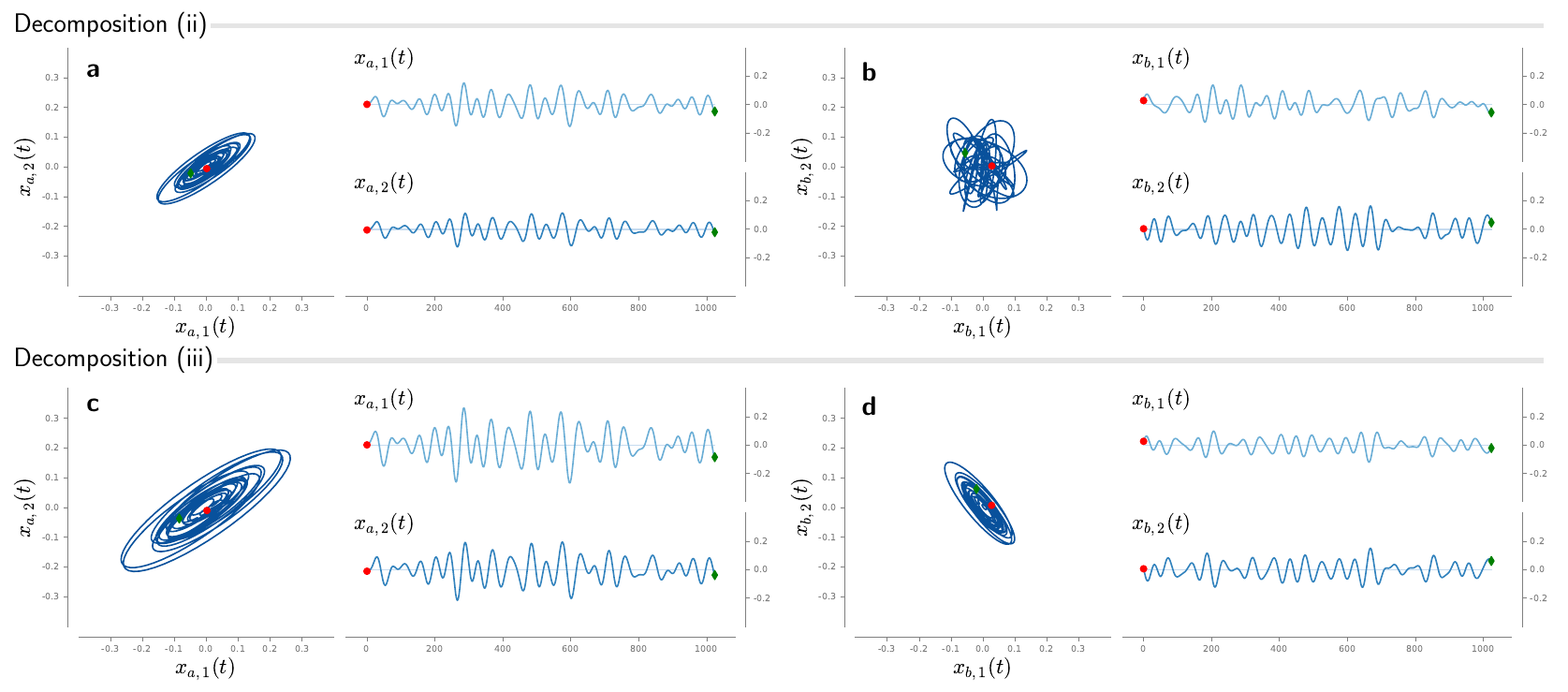}
    \caption{
    Decompositions (ii) and (iii) of the bivariate signal of Fig. \ref{fig:experimental}a.
    See Table \ref{table:decompositionBivSig} for expressions.
    \textbf{\textsf{a}}\: polarized part and
    \textbf{\textsf{b}}\: unpolarized part of decomposition (ii).
    Components are correlated.
    \textbf{\textsf{c}} and
    \textbf{\textsf{d}}: uncorrelated, orthogonal polarized parts of the original signal obtained thanks to decomposition (iii).}
    \label{fig:decompositions}
  \end{figure*}

\section{Conclusion}
\label{sec:conclusion}

This paper provides a complete and powerful framework for linear time-invariant filtering of bivariate signals.
The proposed framework yields a direct description of filtering in terms of physical quantities borrowed from polarization optics.
Our formalism reveals the specifity of bivariate signals and is crucial to the physical understanding of even basic operations such as linear filtering.
The natural expression of each filter directly in terms of eigenproperties and relevant physical parameters simplifies modeling, design, calculations and interpretations.
By studying in detail the two types of filters called unitary and Hermitian filters, we have also been able to give strong physical interpretations in terms of \emph{birefringence} or \emph{diattenuation} effects.

We have emphasized the relevance of our work on three fundamental applications of signal processing.
A spectral synthesis method to simulate any Gaussian stationary random bivariate signal with desired spectral and polarization properties has been presented.
It has been shown that the Wiener denoising problem can be efficiently designed in the quaternion domain, leading to new interpretations for the bivariate case.
Original decompositions of bivariate signals into two parts with specific properties have been studied.
Our approach paves the way to further developments in estimation, simulation and modelling of bivariate signals.
The approach is numerically efficient and relies on the use of FFT.
An open-source implementation of the presented framework will be soon available in the Python companion package
\texttt{BiSPy}\footnote{Documentation available at \texttt{\url{https://bispy.readthedocs.io/}}}.

\appendices

\section{Linear algebra and quaternion equivalence} 
\label{sec:app:linear_algebra_and_quaternion_equivalence}

\subsection{Matrix-vector and quaternion operations} 
\label{sub:general_matrix_quaternion_equivalence}
Eq. (\ref{eq:spectralVectorQuaternionRep}) shows that quaternions can be represented as complex $\bbCj$-vectors.
Let $\ve{X} = [X_1, X_2]^T$ and $\ve{Y} = [Y_1, Y_2]^T$ complex $\bbCj$-vectors corresponding to quaternions $X$ and $Y$.
Let $\mat{M}$ denote an arbitrary complex 2-by-2 matrix.
The matrix-vector relation $\ve{Y} = \mat{M}\ve{X}$ describes an arbitrary linear transform of $\bbCj^2$.

To obtain the corresponding relation between quaternions $Y$ and $X$, write explicitly the matrix-vector relation
\begin{equation}
    \begin{pmatrix}
        Y_1\\
        Y_2
    \end{pmatrix}
    =
    \begin{pmatrix}
        a & b\\
        c & d
    \end{pmatrix}
    \begin{pmatrix}
        X_1\\
        X_2
    \end{pmatrix}
    =
    \begin{pmatrix}
        a X_1 + b X_2\\
        c X_1 + d X_2
    \end{pmatrix}\label{eq:app_matrix_vector}
\end{equation}
where $a, b, c, d \in \bbCj$. Using (\ref{eq:spectralVectorQuaternionRep}) and
that for any $q = q_1 + \bmi q_2 \in \bbH$, $q_1, q_2 \in \bbCj$ one has $q_1 = (q+\involj{q})/2$ and $\bmi q_2 = (q-\involj{q})/2$:
\begin{align}
    Y &= Y_1 + \bmi Y_2  = a X_1 + b X_2 + \bmi\left(c X_1 + d X_2\right)\nonumber\\
      &= \frac{1}{2}\left(a - b\bmi + \bmi c -\bmi d \bmi\right)X\nonumber\\
      &-\frac{1}{2}\left(a + b\bmi + \bmi c + \bmi d \bmi\right)\bmj X\bmj \label{eq:app_general_filter}.
\end{align}
Eq. (\ref{eq:app_general_filter}) is the quaternion domain representation of a generic linear transform of vectors of $\bbCj^2$.

\subsection{Unitary transforms} 
\label{sub:app:unitary_transforms}

Let $\mat{U} \in \Utwo \subset \matrixspace{\bbCj}{2}$,  \ie such that $\mat{U}\mat{U}^* = \mat{U}^*\mat{U} = \mat{I}_2$.
Remark that $\mat{U} = \mat{\tilde{U}}\det(\mat{U})$ where $\mat{\tilde{U}} \in \SUtwo$ and $\det \mat{U} = \exp(\bmj\phi) \in \bbCj$.

Using notations from (\ref{eq:app_matrix_vector}), the matrix $\mat{\tilde{U}}$ is characterized by $d = \overline{a}$, $c = - \overline{b}$ and $\vert a\vert^2 + \vert b\vert^2 = 1$. Thus (\ref{eq:app_general_filter}) simplifies as
\begin{equation}\label{eq:app_su2filter}
    Y = (a - b\bmi) X = \exp(\bmmu \alpha) X.
\end{equation}
Since $\vert a\vert^2 + \vert b\vert^2 = 1$, $a-b\bmi$ is a unit quaternion which can be reparameterized in polar form by its axis $\bmmu$ and angle $\alpha$ such that
\begin{align}
    \bmmu &= \frac{-\bmi \mathrm{Re}\:b + \bmj \mathrm{Im}_{\bmj} a + \bmk \mathrm{Im}_{\bmj}b}{\vert -\bmi \mathrm{Re}\:b + \bmj \mathrm{Im}_{\bmj} a + \bmk \mathrm{Im}_{\bmj}b\vert },\\
    \alpha &= \arccos \mathrm{Re}\:a
\end{align}
Back to $\mat{U} \in \Utwo$, remark that
  \begin{equation}
      \ve{Y} = \mat{U}\ve{X} =  \mat{\tilde{U}}\begin{bmatrix}X_1 e^{\bmj\phi}\\X_2 e^{\bmj\phi}\end{bmatrix},
  \end{equation}
so that replacing $X$ by the quaternion $Xe^{\bmj\phi}$ in (\ref{eq:app_su2filter}) yields,
\begin{equation}
      \text{For } \mat{U} \in \Utwo, \: \ve{Y} = \mat{U}\ve{X} \Longleftrightarrow Y = e^{\bmmu \theta}Xe^{\bmj\phi}.
\end{equation}


\subsection{Hermitian transforms} 
\label{sub:app:hermitian_transforms}

Let $\mat{H}$ be Hermitian, \ie such that $\mat{H}^\dagger = \mat{H}$.
Using notations from (\ref{eq:app_matrix_vector}) one has $a, d \in \bbR$ and $c = -\overline{b} \in \bbCj$. Positive semidefiniteness is given by Sylvester Criterion: $a \geq 0 \quad ad - \vert b\vert^2 \geq 0$,
which also implies that $d \geq 0$.
Eq. (\ref{eq:app_general_filter}) becomes
\begin{equation}
    Y = \frac{1}{2}\left(a + d\right)X - \frac{1}{2}\left( 2b\bmk + (a-d)\bmj\right)X\bmj
\end{equation}
which can be reparameterized such as
\begin{align}
    K &= \frac{a+d}{2} \in \bbR^+\\
    \bmmu &= \frac{(a-d)\bmj + 2b\bmk}{\left[(a-d)^2 + 4\vert b\vert^2\right]^{1/2}}, \: \bmmu^2 = -1\\
    \eta & = \frac{\left[(a-d)^2 + 4\vert b\vert^2\right]^{1/2}}{a + d} \in [0, 1]
\end{align}
Respective domains of $K, \bmmu, \eta$ ensure that the change of variable defines a valid one-to-one mapping. Finally, the input-output relation reads
\begin{equation}
    Y = K\left(X - \eta \bmmu X \bmj \right).
\end{equation}
Parameters $K$ and $\eta$ can be expressed in terms of eigenvalues $\lambda_1, \lambda_2$ ($\lambda_1 \geq \lambda_2 \geq 0)$ of the matrix $\mat{M}$:
\begin{equation}
    K = \frac{\lambda_1 + \lambda_2}{2} \quad \eta = \frac{\lambda_1-\lambda_2}{\lambda_1 + \lambda_2}.
\end{equation}

\section{Wiener filter derivation}
\label{app:wienerFilter}

We keep notations from Section \ref{sub:wiener_denoising}.
Let $\ve{y}(t)$, $\ve{\hat{x}(t)}, \ve{x}(t)$ denote vector representations of quaternions signals $y(t)$, $\hat{x}(t)$ and $x(t)$.
Remark that (\ref{eq:wienerProblemquaternion}) is equivalent to its vector form:
\begin{equation}\label{eq:wienerProblemvector}
  \mathrm{min}\: \Expe{\Vert \ve{\hat{x}}(t) - \ve{x}(t)\Vert^2},
\end{equation}
where $\Vert \cdot \Vert$ is the Euclidean norm of $\bbCj^2$.
The solution to (\ref{eq:wienerProblemvector}) in the Fourier domain is well known \cite{Schreier}
\begin{equation}\label{eq:WienerFilterMatrixGeneralSetup}
  \ve{\hat{X}}(\nu) = \mat{P}_{\ve{x}\ve{y}}(\nu)\mat{P}_{\ve{y}\ve{y}}^{-1}(\nu)\ve{Y}(\nu)
\end{equation}
where $\mat{P}_{\ve{x}\ve{y}}(\nu), \mat{P}_{\ve{y}\ve{y}}(\nu)$ are the usual (cross-) spectral density matrices of $\ve{x}(t), \ve{y}(t)$, respectively.
The Wiener filter for the denoising problem (\ref{eq:denoising_quaternion_problem}) is
\begin{equation}\label{eq:WienerFilterMatrixAdditiveModel}
  \ve{\hat{X}}(\nu) = \mat{P}_{\ve{x}\ve{x}}(\nu)\mat{P}_{\ve{y}\ve{y}}^{-1}(\nu)\ve{Y}(\nu)
\end{equation}
Eq. (\ref{eq:WienerFilterMatrixAdditiveModel}) shows that $\ve{\hat{X}}(\nu)$ is obtained from $\ve{Y}(\nu)$ by 2 successive Hermitian filters, since spectral density matrices are Hermitian -- and so are their sum and inverse.
Introducing an intermediate variable $\ve{Z}$ one gets
\begin{align}
    \ve{Z}(\nu) &= \mat{P}_{\ve{y}\ve{y}}^{-1}(\nu) \ve{Y}(\nu)\\
     \ve{\hat{X}}(\nu) &= \mat{P}_{\ve{x}\ve{x}}(\nu) \ve{Z}(\nu)
\end{align}

Quaternions equivalents are readily obtained using (\ref{eq:app_general_filter}) and definitions of matrix spectral densities in terms of Stokes parameters $S_i$, $i = 0, 1, 2, 3$ \cite[p. 214]{Schreier}:
\begin{align}
   Z(\nu) &= 2\left[(1-\Phi^2_y(\nu))S_{0, y}(\nu)\right]^{-1}  \nonumber\\
    &\qquad \times\left(Y(\nu) + \Phi_y(\nu)\bmmu_y(\nu)Y(\nu)\bmj\right)\label{eq:app:first_filterQuatWiener}\\
   \hat{X}(\nu) &= 2^{-1}S_{0, x}(\nu)\left(Z(\nu)- \bmmu_x(\nu)\Phi_x(\nu)Z(\nu)\bmj\right)\label{eq:app:second_filterQuatWiener}
\end{align}
since Stokes parameters and polarization axis are related like \cite{flamant2017spectral}
$S_0\Phi\bmmu = \bmi S_3 + \bmj S_1 + \bmk S_2$.
Plugging (\ref{eq:app:first_filterQuatWiener}) into (\ref{eq:app:second_filterQuatWiener}) and reorganizing terms yields to the general Wiener filter expression (\ref{eq:wienerFilterquaternionGeneric}).
To obtain the error expression remark that \cite[Theorem 1]{flamant2017spectral}
\begin{equation}\label{appeq:errorWiener}
  \epsilon = \int_{-\infty}^\infty \mathcal{S}(\Gamma_{ee}(\nu))\mathrm{d}\nu
\end{equation}
where $e(t) = \hat{x}(t) - x(t)$.
Using the spectral density definition (\ref{appeq:quaternionSpectralDensity}) together with the Wiener filter expression (\ref{eq:wienerFilterquaternionGeneric}) one gets the optimal error expression (\ref{eq:optimalErrorSigObs}) by developing (\ref{appeq:errorWiener}).
To obtain (\ref{eq:optimalErrorSigNoise}) start by writing explicitly $\Gamma_{yy}(\nu) = \Gamma_{xx}(\nu)+\Gamma_{ww}(\nu)$ such that ($\nu$-dependence omitted):
\begin{align}
  \Gamma_{yy} &= S_{0, x} + S_{0, w} + S_{0, x}\Phi_x\bmmu_x + S_{0, w}\Phi_w\bmmu_w\\
  &= S_{0, y}[1 + \Phi_y\bmmu_y],
\end{align}
where, using $\alpha = S_{0, x}/S_{0, w}$ the frequency domain SNR:
\begin{align}
  S_{0, y} &= S_{0, x} + S_{0, w}\label{appeq:SOy1}\\
  \Phi_y\bmmu_y &= \frac{\alpha}{1+\alpha}\Phi_x\bmmu_x + \frac{1}{\alpha +1}\Phi_w\bmmu_w\label{appeq:SOy2}.
\end{align}
Plugging (\ref{appeq:SOy1}) and (\ref{appeq:SOy2}) into (\ref{eq:optimalErrorSigObs}) yields (\ref{eq:optimalErrorSigNoise}).

\section{Simulation of bivariate white noise}
\label{app:whitenoise}
For sake of completeness we recall some recent results from \cite{flamant2017spectral}.
A bivariate white noise $w(t) = u(t) + \bmi v(t)$ has a constant spectral density given by
\begin{equation}\label{appeq:spectralDensWhiteNoise}
  \Gamma_{ww}(\nu) = \sigma_u^2+\sigma_v^2 + \bmj (\sigma_u^2-\sigma_v^2) + 2\bmk \rho_{uv}\sigma_u\sigma_v.
\end{equation}
where $\sigma_u^2, \sigma_v^2$ are variances of white noises $u(t)$ and $v(t)$, and $\rho_{uv}$ is the correlation between $u(t)$ and $v(t)$.
This spectral density has no $\bmi$-component, meaning that a bivariate white noise is \emph{always} partially linearly polarized.
Importantly, $w(t)$ is unpolarized when $\sigma_u^2 = \sigma_v^2$ and $\rho_{uv} = 0$, \ie when $w(t)$ is \emph{proper} \cite{Picinbono1997a}.

Simulating a bivariate white noise $w(t)$ is equivalent to simulating 2 correlated real white noises $u(t)$ and $v(t)$.
Alternatively \cite{flamant2017spectral}, one can simulate $w(t)$ directly with the desired polarization properties using an unpolarized/polarized parts decomposition.
Let $ 0 \leq \Phi\leq 1$ be the desired degree of polarization, and $\theta \in [-\pi/2, \pi/2]$ the linear polarization orientation angle and $S_{0,w} > 0$ the total power.
Let $w^{\texttt{u}}(t)$ be an unpolarized white noise and $w^{\texttt{p}}(t)$ be a real-valued white noise, both of unit variance and independent from each other.
Then the white noise $w(t)$ constructed as
\begin{equation}\label{eq:UPdecompWN}
  w(t) = \sqrt{1-\Phi}\sqrt{S_{0, w}}w^{\texttt{u}}(t) + \sqrt{\Phi}\sqrt{S_{0, w}}\exp(\bmi \theta)w^{\texttt{p}}(t)
\end{equation}
has spectral density $\Gamma_{ww}(\nu) = S_{0, w} + \bmj \Phi S_{0, w} \cos2\theta + \bmk\Phi S_0\sin2\theta$ where one recognizes a linear polarization state with spherical coordinates $(\Phi, 2\theta, 0)$, see Fig. \ref{fig:poincaresphere}.

\section{Spectral representation of stationary bivariate signals}
\label{app:spectralRep}
We recall some important results from \cite{flamant2017spectral}.
When $x(t)$ is a random bivariate signal the QFT definition (\ref{eq:definitionQFT}) is no longer valid.
Instead it has to be replaced with the spectral representation theorem \cite[Theorem 1]{flamant2017spectral} which states for harmonizable signals $x(t)$ there exist spectral increments $\mathrm{d}X(\nu)$ such that
\begin{equation}
  x(t) = \int_{-\infty}^{+\infty} \mathrm{d}X(\nu)e^{\bmj2\pi\nu t},
\end{equation}
the equality being in the mean-square sense.
Then one defines the quaternion spectral density $\Gamma_{xx}(\nu)$ accordingly \cite{flamant2017spectral} as
\begin{equation}\label{appeq:quaternionSpectralDensity}
  \Gamma_{xx}(\nu)\mathrm{d}\nu = \Expe{\vert \mathrm{d}X(\nu)\vert^2} + \Expe{\mathrm{d}X(\nu)\bmj\overline{\mathrm{d}X(\nu)}}
\end{equation}
where $\Expe{\cdot}$ denotes the mathematical expectation.

Let $x(t)$ and $y(t)$ be two jointly stationary bivariate signals.
These signals are uncorrelated \cite{flamant2017spectral} if and only if, for all $\nu$
\begin{equation}
  \Expe{\mathrm{d}X(\nu)\overline{\mathrm{d}Y(\nu)}} =
  \Expe{\mathrm{d}X(\nu)\bmj\overline{\mathrm{d}Y(\nu)}} = 0.
\end{equation}
This is the quaternion equivalent to saying that the cross-spectral density matrix is zero: $\mat{P}_{xy}(\nu) = 0$.

\bibliography{references}
\bibliographystyle{IEEEtran}

\end{document}